\documentclass[letterpaper, 10 pt, conference]{ieeeconf}  

\IEEEoverridecommandlockouts                              

\usepackage{graphics} 
\usepackage{epsfig} 

\usepackage{times} 

\usepackage{amsmath}

\usepackage{amssymb}  
\usepackage{comment}

\usepackage{enumerate}
\usepackage{color}
\usepackage[hidelinks]{hyperref}
\usepackage[compress]{cite}

\newtheorem{remark}{Remark}
\newtheorem{pro}{Proposition}

\newtheorem{thm}{Theorem}
\newtheorem{assumption}{Assumption}
\newtheorem{definition}{Definition}
\newtheorem{example}{Example}

\ifodd 1

\newcommand{\com}[1]{{\color{red}{Comment: #1}}}
\else
\newcommand{\com}[1]{}
\fi

\title{\LARGE \bf
Incentive Design in a Distributed Problem with Strategic Agents}

\author{Donya Ghavidel, Pratyush Chakraborty, Enrique Baeyens, Vijay Gupta, and Pramod P. Khargonekar
\thanks{Donya Ghavidel and Vijay Gupta are with the Department of Electrical Engineering,
        University of Notre Dam, IN, USA
        {\tt\footnotesize dghavide@nd.edu , vgupta2@nd.edu}}
\thanks{Pratyush Chakraborty is with the Department of Mechanical Engineering,
        University of California, Berkeley, CA, USA
        {\tt\footnotesize pchakraborty@berkeley.edu}}%
\thanks{Enrique Baeyens is with Instituto de las Tecnolog\'{\i}as Avanzadas de la Producci\'on,
Universidad de Valladolid, Valladolid, Spain
        {\tt\footnotesize enrbae@eis.uva.es}}%
\thanks{Pramod P. Khargonekar is with the Department of Electrical Engineering and Computer Science,
        University of California, Irvine, CA, USA
        {\tt\footnotesize pramod.khargonekar@uci.edu}}%
}

\begin{document}
\maketitle
\thispagestyle{empty}
\pagestyle{empty}

\begin{abstract}
 In this paper, we consider a general distributed system with multiple agents who select and then implement actions in the system. The system has an operator with a centralized objective. The agents, on the other hand, are self-interested and strategic in the sense that each agent optimizes its own individual objective. 
 The operator aims to mitigate this misalignment by designing an incentive scheme for the agents. The problem is difficult due to the cost functions of the agents being coupled, the objective of the operator not being social welfare, and the operator having no direct control over actions being implemented by the agents. This problem has been studied in many fields, particularly in mechanism design and cost allocation. However, mechanism design typically assumes that the objective of the operator is social welfare and the actions being implemented by the operator. On the other hand, cost allocation classically assumes that agents do not anticipate the effect of their actions on the incentive that they obtain. We remove these assumptions and present an incentive rule for this setup by bridging the gap between mechanism design and classical cost allocation. We analyze whether the proposed design satisfies various desirable properties such as social optimality, budget balance, participation constraint, and so on. We also analyze which of these properties can be  satisfied if the assumptions on cost functions of the agents being private and the agents being anticipatory are relaxed.
 \end{abstract}

\section{Introduction}

Using game theoretic tools to design distributed controllers in multi-agent systems where each agent has an individual utility function is a well-established research field~\cite{marden2012,atzeni2}. Many impressive results are available for questions such as analyzing the behavior of agents, for instance in a Nash equilibrium sense \cite{bacsar1998}, using learning algorithms to identify these strategies when utility functions for the other agents may not be known \cite{konstant}, and assigning utility functions to cooperative agents so that a desired behavior emerges in an optimization problem \cite{waslander,ansaripour2013robust,arslan2006}. Applying results from game theory to control raises several new challenges including presence of dynamics, constraints, and utility functions that are non-standard in game theory. 

In this paper, we are interested in a set-up where multiple agents that have their own utility functions act on a single system. Further, there is a central operator that has its own utility function and can incentivize (or tax) the agents to behave in a manner it desires~\cite{butkovskiui1969, ansaripour2018multi}. Examples of problems in this framework may include design of resource allocation algorithms in shared infrastructure systems, for instance, routing and toll design in transportation networks~\cite{vetta}, transmission network design and operation in power grids \cite{ruiz2007}, and so on. The central operator does not have knowledge of the utility functions of the agents. Since the agents are strategic, they may misreport these functions upon being asked. In such a setup, questions about incentive design for various aims such as optimizing the system efficiency, fairness, budget balance, and so on naturally arise. 

The fields that are primarily relevant to this question are mechanism design~\cite{mas1995} and cost allocation~\cite{young1985cost}. Mechanism design considers the question of incentive (or tax) design when utility functions of the agents are unknown to the central operator; the classical application is in auction theory \cite{krishna}. Three characteristics complicate the application of classical mechanism design to the distributed control setup considered in this paper:
\begin{enumerate} 
\item\label{item:first} Agents have interdependent utility functions and thus the utility function of each agent (and thus its desired strategy) is a function of the strategy of the other agents.
\item\label{item:second} The utility function of the operator may not be the sum of the utility functions of the agents, which is usually termed social welfare.
\item\label{item:third} The actions are implemented by the agents and the operator does not have direct control over them.
\end{enumerate}
There have been some recent works that have considered some of these issues.  In \cite{tanaka}, the authors present a mechanism for a group of self-interested power generators to implement socially optimal model predictive control algorithms for load frequency control (LFC). This work presents  a mechanism to mitigate the information asymmetry between the agents and the operator with the  application of MPC-based LFC, but ignores issues~(\ref{item:second}) and~(\ref{item:third}) mentioned above. A mechanism to execute a distributed algorithm  by strategic agents, who have individual utility functions and the operator does not have knowledge of the utility functions, is designed in \cite{tanaka2}. However, this work  ignores issues~(\ref{item:second}) and~(\ref{item:third}). A mechanism to achieve social welfare maximization and budget balance by agents in a LQG problem, formulated for power networks, is presented in~\cite{okajima}; however, once again, these works ignore issues~(\ref{item:first}) and~(\ref{item:second}). Finally, we would like to mention the work~\cite{tanakadynamic} that develops a mechanism to incentivize agents to implement decisions deemed optimal by the central operator. However, that work assumes the utility functions of the agents to be known by the central operator.

Cost allocation classically considers the problem of allocating to various agents the loss of the utility function of the central operator due to the self-interested behavior of the agents. Many allocation rules, such as the proportional rule, the Shapley value, and the nucleolus, that each display different sets of properties have been proposed in the literature. Works that develop these roles for control applications include, e.g.,~\cite{Bitar, chakraborty2016 }. In particular,\cite{ Bitar} proposes a cost allocation proportional to the marginal contributions of the agents while satisfying some fairness properties. In~\cite{chakraborty2, chakraborty2016} the authors use cost allocation based on standalone cost principle and cost causation principle respectively. A different, but related, line of work is that in~\cite{lian,an20} that develops a distributed algorithm based on the Nash Bargaining Solution to fairly allocate the costs of the communication infrastructure in  \cite{lian} and power management in  \cite{an20} among the agents. However, to the best of our knowledge, the literature on cost allocation assumes the agents to be not cost anticipatory, in the sense that they do not optimize the decision that they take with respect to both their a priori utility function and the cost that will be allocated to them.


Thus, there is a gap in the literature. While mechanism design considers strategic agents with private utility functions that are unknown to the central operator, existing work does not consider issues such as the utility functions of the agents being coupled, the operator being interested in a specific objective that is not the social welfare, and actions being implemented by the agents and not the operator all together. On the other hand, while cost allocation can consider coupled utility functions and actions being implemented by the agents, it does not consider agents that anticipate the cost allocation rule and an operator that has a specific utility function. The main contribution of this paper is to fill this gap in the context of a distributed multi-agent control problem.

Specifically, we consider a multi-agent problem in which the central operator and the agents have private individual utility functions. The agents are selfish in that they seek to optimize their own utility, strategic in that they can misreport any information, and anticipatory in the sense that they consider the impact of any incentive design from the operator. All the issues~(\ref{item:first})--(\ref{item:third}) are present. We propose an incentive design to align the actions of the individual agents with those desired by the central operator. We compare this design (that assumes utility functions to be private information) with an mechanism design (that assumes knowledge of the utility functions of the agents) inspired by traditional mechanism design and a cost allocation for the case when agents are not anticipatory, with respect to various properties such as budget balance, participation constraint and so on.

The rest of the paper is organized as follows. In Section \ref{PF},
the problem statement and definition of desired properties are presented. In Section \ref{MR}, we propose
an incentive mechanism satisfying the properties, compare it with the classic cost allocation rule, and interpret it as a mechanism design. In
Section \ref{conclusion}, we conclude the paper and present some avenues
for future work.

\section{ Problem Formulation}
\label{PF}
We begin with the system model and next define the desired properties.
\subsection{System Model}
We consider a general distributed system with multiple agents who decide their actions according to individual utility functions and a central operator that also has an associated objective function. Let $\mathcal N=\{1,\cdots, N\}$ denote the set of the $N$ agents present. Let $u_{i}$ be the action decided by the $i$-th agent. Denote the set of actions of all the agents by $U\triangleq\{u_1,\cdots, u_N\}$ and that of actions of all the agents except the $i$-th agent by $u_{-i}\triangleq\{u_1,\cdots,u_{i-1},u_{i+1},\cdots, u_N\}$. In keeping with the interpretation as a control problem, the utility for each agent and the central operator is provided through a cost function that needs to be minimized. Specifically, denote by $C_{i}(U)$ the cost function of the $i$-th agent and by $J(U)$ that of the central operator. Notice that all the cost functions depend on the actions of every agent; to emphasize this, we sometimes use the alternative notation $C_{i}(u_{i},u_{-i})$. Denote the control actions desired by the central operator from the $i$-th agent by $u_{i}^{\star}.$ In other words,
\begin{equation}
U^\star=\arg\min J(U),
\label{eq1}
\end{equation}
where $U^{\star}\triangleq\{u_1^{\star},\cdots, u_N^{\star}\}.$
 
Since the cost function of each agent depends on the actions of all the agents, there is a non-cooperative game among them. The best response strategy of the $i$-th agent can be defined as 
\begin{equation}
 \bar{u}_i (u_{-i})=\arg \min C_i(U).
 \label{eq2}
\end{equation}
Further, a Nash Equilibrium (NE) of the game can be defined as the set $\bar{U}\triangleq\{\bar{u_1},\cdots, \bar{u}_{N}\}$ such that the following inequalities hold:
\begin{equation}
C_i(\bar{u}_i, \bar{u}_{-i})\leq C_i(u_i, \bar{u}_{-i})\quad \text{for all} \quad i.
\label{eq3}
\end{equation}
We make the following assumption:
\begin{assumption}
At least one pure strategy Nash equilibrium as defined by the set $\bar{U}$ exists.
\end{assumption}
\begin{remark}
Note that the Nash equilibrium may not be unique. 
\end{remark}

Since the set of actions desired by the central operator $U^{\star}$ may not be a Nash equilibrium, the agents will not, in general, choose actions that are desired by the central operator. To incentivize the agents to choose actions that are desired by himself, the operator assigns an incentive calculated according to the function $t_i(U)$ to the $i$-th agent.
Considering the cost minimization set-up, this incentive can correspond to a reward (if negative) or a tax (if positive) on the $i$-th agent. If the coupled decision problem outlined above corresponds to a resource allocation problem in a shared infrastructure system with the central operator interested in minimizing the cost to society, the incentive can be interpreted as a fee that each agent pays to use the system with the magnitude being different for each agent according to how much inefficiency they cause with respect to the social cost.

\begin{remark}
We allow an agent to withdraw from the incentive scheme if it so desires. In other words, the participation of the agents is voluntary and the so-called outside option of the agent $i$ is simply to opt out unilaterally out of the incentive design.
\end{remark}
With a given incentive design, the cost functions of the agents and the central operator change. Denote by $\mathcal{O}$ the set of agents which opt out. Then, the cost function of the central operator is given by $J(U)-\sum_{i\in\mathcal{N}\setminus\mathcal{O}}t_{i}(U),$ that of any agent $i\in\mathcal{O}$ by $C_{i}(U),$ and that of any agent $i\in\mathcal{N}\setminus\mathcal{O}$ by $C_{i}(U)+t_{i}(U).$ We also introduce some notation to distinguish between the Nash equilibria that occur when agent $i$ participates versus when it opts out. Denote the Nash equilibrium strategies of the agents when agent $i$ participates by ${U}'$ or $({u}'_i, {u}'_{-i}),$ and for the case when agent $i$ unilaterally opts out of the mechanism by  $(\tilde{u}_i, \hat{u}_{-i})$. The notation for the various cases is summarized in Table \ref{table2}. 
We make the following assumption.

\begin{table}
\vspace{3mm}
\caption{Notation summary in various scenarios}
    \begin{tabular}{ p{3cm} | p{1cm}| p{3.1cm}}
    \hline
    \textbf{Scenarios} & NE & cost  of the operator \\ \hline
    \vspace{0.1mm} With no incentive design & \vspace{0.1mm}$(\bar{u}_i, \bar{u}_{-i})$ &  \vspace{0.1mm}$J(\bar{U})$  \\ \hline
    With incentive design when agent $i$ participates &$({u}'_i, {u}'_{-i})$ & $J(U')-\sum_{j\in\mathcal{N}\setminus\mathcal{O}}t_{j}(U')$  \\ \hline
    With incentive design when agent $i$ opts out & $(\tilde{u}_i, \hat{u}_{-i})$ & $J(\tilde{u}_i, \hat{u}_{-i})-\sum_{j\in\mathcal{N}\setminus\mathcal{O}}t_{j}(\tilde{u}_i, \hat{u}_{-i})$\\
    \hline
    \end{tabular}
    \label{table2}
    \vspace{-6mm}
\end{table}

\begin{assumption}
We assume that at least one pure strategy Nash equilibrium of the form $({u}'_i, {u}'_{-i})$  and $(\tilde{u}_i, \hat{u}_{-i})$ exist for the case when any agent $i$ that decides to opt out.
\end{assumption}
The problem considered in the paper is to analyze incentive scheme that satisfies certain desirable properties with or without access to the cost functions of the agents. Next, we define the set of desirable properties.

\subsection{Properties}
\label{sec:prop}
We are interested in incentive schemes that satisfy certain properties. To define the properties, it will be useful to define two costs generated by strategic behavior of the agents. 
\begin{definition}[Excess cost]
For a given set $U_r\triangleq\{{u_{r_1}},\cdots,{u}_{r_N}\}$ of actions realized by the agents, the excess cost imposed on the system operator is defined by
\[
\Theta=J( U_r)-J(U^\star).
\]
\end{definition}
Note that depending on the agents participation, $U_r$ can be equal to $U'$, $\bar{U}$ and $(\tilde{u}_i, \hat{u}_{-i})$.

 \begin{definition}[Marginal cost]
For an action $u_{r_{i}}$ by the $i$-th agent, consider the set of actions $U^{i}_m\triangleq\{u^\star_1, \cdots, u^\star_{i-1}, u_{r_i}, u^\star_{i+1}, \cdots, u^\star_{N}\},$ where  $u_j^\star$ denotes the action of the $j$-th agent that is optimal according to the cost function $J$ of the operator. Then, the  marginal cost $\theta_{i}$ is defined as  $\theta_{i}=J(U^{i}_m)-J(U^\star)$.

\end{definition}
Note that both $\theta_{i}\geq 0$ and $\Theta\geq 0.$ Further, the excess cost is generated because of the marginal costs for all the agents.

We now list the properties we are interested in~\cite{young1985cost,mas1995}.
\begin{enumerate}
\item \emph{Social optimality:} An incentive scheme is said to be socially optimal if the actions taken by the agents at Nash equilibrium with the incentive scheme in place are the same as the actions that optimize the cost function of the operator, i.e., if  $U'=U^\star$.
\item \emph{Budget balance:}
An incentive scheme is said to be  (resp. weak) budget balanced, if at the resulting Nash equilbrium, the total incentive $t_i$ is equal to (resp. not less than) the excess cost. Thus,  an incentive scheme is  budget balanced if $\sum_{i=1}^{N}t_i = \Theta$
and weak budget balanced if $
\sum_{i=1}^{N}t_i > \Theta.$

\item \emph{Participation constraint}: Participation constraint implies that the agents benefit from participation in the incentive scheme in the sense that for every agent $i$, the cost when it opts out of the incentive design is larger than its cost function when it participates in the incentive design. Thus, an incentive scheme satisfies participation constraint for agent $i$ if 
$C_i(\tilde{u}_i, \hat{u}_{-i})\geq C_i(U')+t_i.$


\item \emph{Equity:} An incentive design is said to satisfy equity if for any two agents $i$ and $j$ that have the same marginal costs, the incentive is equal as well. Thus, an incentive design satisfies equity if $\forall i,j$ such that $\theta_i = \theta_j$, it holds that $t_i = t_j$.

\item \emph{Monotonicity:} An incentive design is said to be monotonic if every agent $i$ that has a higher marginal cost than agent $j$ is assigned a higher incentive (or tax) than agent $j$. Thus, an incentive design is monotonic if for any $i,j$ such that $\theta_i \geq \theta_j$, it holds that $t_i \geq t_j$.
\end{enumerate}


\vspace{-1.5mm}

\subsection{Problem Considered}
Given the above setup, we are interested in designing the incentive scheme that satisfy the properties mentioned in Section~\ref{sec:prop}. 
In some sense, such an incentive design presents a bridge between mechanism design and cost allocation. While similar to mechanism design we consider anticipatory strategic agents, we assume that their cost functions are coupled and unknown to the central operator who is interested in a specific cost function that is not the social welfare, and actions being implemented by the agents and not the operator. On the other hand, while similar to cost allocation we consider coupled cost functions and actions being implemented by the agents, we assume that the agents anticipate the cost allocation rule. In Section~\ref{CA}, we consider anticipatory agents as well as  non-anticipatory agents  and study the properties achieved by the scheme in this scenario that is inspired by classic cost allocation. Similarly, in Section~\ref{PD}, we relax the assumption of cost functions being unknown to the central operator and study the properties that can be achieved by a mechanism inspired by the classic Vickrey-Clarke-Groves (VCG) mechanism ~\cite[Chapter~23]{mas1995}. Finally, we  compare  the proposed  incentive structures in terms of properties that can be achieved by each.
\section{Main Results}
\label{MR}
We begin by presenting the incentive design for the problem considered above. We will then study the incentive with assumptions similar to classic mechanism design and cost allocation.

\subsection{Proposed Incentive Design}
\label{CAW}
We consider a set-up in which the operator does not access  the knowledge of the cost function of each agent  $C_i$; rather, the knowledge of the resulting actions of each agent suffices.  

We would like to note that the agents are anticipatory. Thus, they know about the incentive $t_i$ and optimize their decisions after including  $t_i$ into their cost function when selecting their decisions. That the incentive design problem is not trivial as can be seen from the following, perhaps counter intuitive, result. One may imagine that since one source of the complexity in the problem is the coupling of cost functions of the various agents, removing that coupling may imply that one or more properties are always satisfied. However, the following result  shows that if the cost functions are decoupled so that $C_i=C_i(u_i)$ for all $i$,  the  participation constraint and budget balance can not hold simultaneously.
\begin{pro}
\label{pro6}
For $C_i=C_i(u_i)$ for all $i$, if we design an incentive such that the agents rationally participate, the (weak) budget balance constraint can never be satisfied.
\end{pro}
\begin{proof}
See Appendix in \cite{extend}.
\end{proof}
Note that $C_i=C_i(u_i)$ implies that the decisions made by the agents are independent of each other.
Unlike  decoupled case, for the case of coupled  cost  function it is sometimes possible to come up with an incentive which guarantees participation constraint as well as budget balance properties.

\begin{example}
 Let the cost functions of agent 1 and 2 be
\[
C_1(U)=u^2_1-2u_1u_2, \quad
C_2(U)=u_1u_2-u_2.
\]
With no incentive, the Nash equilibrium solution is given by $\bar{U}=(1, 1)$. The cost function of the central operator on the other hand is \[J(U)=(u_1-\frac{3}{4})^2 + (u_2-2)^2.\] Thus, $U^{\star}=(\frac{3}{4}, 2)$  and  $J(U^{\star})=0$.
The operator gives incentive $t_i$ to agent $i$ such that $t_1(u_1)= u_1^2, \quad t_2(u_2)=-\frac{1}{2}.$
The cost functions of the agents after having incentive are $ C_1(U) + t_1(u_1)$ and  $ C_1(U) + t_2(u_2)$. Therefore, Nash equilibrium if the agents participate in the incentivization  is given by
 $  U'=(1,2) $ and, consequently the value of the cost functions and the incentives are given by 
$C_1({U}')=-3$, $C_2({U}')=0$, $t_1({u}'_1)=1 $ and $t_2({u}'_2)=-\frac{1}{2}$.

On the other hand, if agents wish to opt out unilaterally, their Nash equilibrium solution and the  cost functions are as follows.
$(\tilde{u}_1, \hat{u}_{2})=(1,1)$  is calculated as the NE of $C_1$ and $C_{2}+t_{2}$. $(\hat{u}_{1}, \tilde{u}_2)=(1,2)$ is calculated as the NE of $C_2$ and $C_{1}+t_{1}$.
As a result, $C_1(\tilde{u}_1, \hat{u}_{2})=-1$ and $C_2(\hat{u}_{1}, \tilde{u}_2)=0$.
Next, we check the participation constraint for $i=1,2$ as
$
C_i(\tilde{u}_i, \hat{u}_{-i})\geq C_i(U')+t_i,
$
which holds for both of the agents using the designed incentive.
Since participation constraint is satisfied for both of the agents, they realize $U'=(1,2)$ 
and $J({U}')- J(U^{\star})=\frac{1}{16}$. Since $t_1(U')+t_2(U')=\frac{1}{2}$ in this case then one can  see that weak budget balance holds $J({U}')- J(U^{\star}) < \sum_{i=1}^2 t_i(U').$

Thus, participation constraint and weak budget balance are satisfied. One can check that monotonicity is also satisfied. Note that using this incentive by the operator, the agents shift from $\bar{U}$ to $U'$. If one calculates the cost function of operator in these two cases $J(\bar{U})=1+\frac{1}{16}$ and $J(U')-(t_1+t_2)=\frac{1}{16}-\frac{1}{2}$. Hence, the cost function of the operator decreased by offering the incentive and the incentive is in the benefit of the operator. 
\end{example}

This is to note that social optimality may not always be satisfied by incentive design.
In general without knowledge of cost function of the agents, one can not come up with an incentive design  which \textit{always} (for all cost function) satisfies social optimality and participation constraint.




 \subsection{Interpretation as Cost Allocation}
\label{CA}
The role of any cost allocation is to suitably allocate the excess cost $\Theta$ of the operator that is generated due to the self-interested behavior of the agents to the agents in a manner that satisfies certain properties. We propose the following {\em proportional cost allocation rule}, under which the incentive $t_i$ for agent $i$ by the operator is given by
 \begin{equation}
 \label{cost}
t_i=\frac{\theta_i}{\sum_{j=1}^{N}\theta_j} \Theta.
\end{equation}
We sometimes use the alternative notation $t_i(U_r, U^\star)$.
We now analyze the properties of the cos allocation scheme~(\ref{cost}) when the underlying framework is that of classical cost allocation. Recall that in classical cost allocation, the agents are assumed to be non-anticipatory in that they do not anticipate the effect of the incentive scheme while choosing their actions. An alternate interpretation is that each agent minimizes its own cost function and incurs a cost for that action in an ex-post manner.

The presence of non-anticipatory agents implies that the participation constraint needs to be modified. Recall that the participation constraint  was of the form $C_i(\tilde{u}_i, \hat{u}_{-i})\geq C_i(U')+t_i(U').$ However, when the agents are non-anticipatory, they optimize their own cost functions irrespective of the incentive design and the participation constraint in this case implies
$C_i(\bar{U})\geq C_i(\bar{U})+t_i,$
 which is never satisfied. Thus the central operator needs to have the authority to penalize the agents so that they participate. To proceed in this case, we define weak participation constraint as being satisfied when the agents have to pay $\theta_i$, if they do not participate in the incentive scheme. Thus, the participation constraint in the weak sense is given  as
\[
t_i +C_i(\bar{U})\leq \theta_i
+C_i(\bar{U}),
\]
or as $t_{i}\leq\theta_i.$ 
Non-anticipatory  agents ignore the cost allocation rule while selecting their actions and do not modify their decision due to the incentive design. Thus, the same actions are chosen irrespective of the design of the incentive scheme $t_i$.   Thus, social optimality is not guaranteed unless the objectives of the agents and the operator are aligned due to some special structure on the cost functions $J$ and $C_{i}$. Since the operator does not require  the knowledge of the cost function of each agent  $C_i$ in implementing~(\ref{cost}), the same incentive design can be used if agents are anticipatory. 

\begin{remark}
For the anticipatory agent, the following is a sufficient condition for the proposed cost function to be socially optimal
\begin{equation*}
\label{Social}
\text{for any $U_r$ and $i\in \mathcal N$}, \quad C_i(U_r)+t_i(U_r,U^\star)\geq C_i(U^\star).   
\end{equation*}
\end{remark}
We make the following assumption that would be crucial to satisfying budget balance and participation constraint simultaneously.
\begin{assumption}
\label{asum1}
The excess cost that needs to be allocated is less than the sum of the individual marginal costs that cause it, i.e., $\Theta\leq \sum_{j=1}^{N}\theta_j$.\end{assumption}

\begin{thm}
\label{thm1}
Consider the proportional allocation incentive design in \eqref{cost}. The design:
\begin{enumerate}[(i)]
 \item satisfies participation constraint for the non-anticipatory agents if Assumption \ref{asum1} holds.
 \item always satisfies  budget balance, equity, and monotonicity.
\end {enumerate}
Further, there is no incentive design $t_i$ that satisfies participation constraint and budget balance simultaneously if agents are non-anticipatory and Assumption \ref{asum1} does not hold.
\end{thm}
\begin{proof}
See Appendix in \cite{extend}.
\end{proof}


 In the following, we provide sufficient conditions on the cost functions $\{J, C_i\}$ for Assumption \ref{asum1} to hold.

\begin{definition}[Separability]
\label{def1}A function $J$ is defined to be separable if it can be written as $J(U)=\sum_{i=1}^{N}f_i(u_i),$
where $f_i(\cdot)$ is a function of $u_i$ but not of $u_{j},$ $j\neq i$.
\end{definition}
\begin{pro}
\label{lem1}
 Assumption \ref{asum1} holds, if one of the following statements  is true.
\begin{enumerate}[(i)]
\item $J$ is separable.
\item  $J$ is of the form \[J({U}, U^\star)=\vert \bar{J}(u_1,\cdots, u_N)-\bar{J}(u^\star_1,\cdots, u^\star_N) \vert,\] where $\bar{J}$ is separable.
\item $J(U^{i}_m)\geq J(\bar U)$ for any agent $i$.
\end{enumerate}
\end{pro}
\begin{proof}
See Appendix in \cite{extend}.
\end{proof}


\begin{example}
Consider a system with two non-anticipatory agents and the cost function of the central operator being given by $J(u_1, u_2)=u^2_1+u^2_2+u_1+u_2-u_1u_2,$ that is minimized by the choice $(u^\star_1,u_2^\star)=(-1,-1)$.
Further, using the definition of the marginal and the excess cost, we have
\[
\theta_1=\bar{u}^2_1-{u^\star}^2_1-u^\star_2(\bar{u}_{1}-u^\star_1), \:\:\:\theta_2=\bar{u}^2_{2}-{u^\star}^2_2-u^\star_{1}(\bar{u}_{2}-u^\star_2)\]
 and $\Theta= \bar{u}^2_{1}-{u^\star}^2_1+\bar{u}^2_{2}-{u^\star}^2_2-\bar{u}_{1}\bar{u}_{2}+u^\star_{1}u^\star_2$.
Assumption \ref{asum1} implies that $\theta_1+\theta_2\geq \Theta,$
which yields
\begin{eqnarray}
-u^\star_2(\bar{u}_1-u^\star_1)-u^\star_1(\bar{u}_2-u^\star_2)\geq -\bar{u}_{2}\bar{u}_1+u^\star_{1}\bar{u}_2 \nonumber
\end{eqnarray}
Given that $(u^\star_1,u_2^\star)=(-1,-1)$, Assumption \ref{asum1} leads to  $(\bar{u}_1+1)(\bar{u}_2+1)\geq 0.$ The feasible region for $\bar{u}_1$ and $\bar{u}_2$  is $\{\bar{u}_1>-1, \bar{u}_2>-1\}$ and $\{\bar{u}_1\leq-1, \bar{u}_2\leq-1\}$. Any $C_1$ and $C_2$ which have a Nash equilibrium solution in the feasible region satisfy Assumption \ref{asum1}.
\end{example}


\subsection{Interpretation as Mechanism Design}
\label{PD}
As discussed earlier, a crucial assumption in our framework is that the cost functions are private knowledge to the agents. If that assumption does not hold,  we can consider the 
%
%
%
following incentive mechanism $t_i$ that is inspired by the VCG mechanism:
 \begin{equation}
 \label{tax}
t_i=C_{r_i}-\hat{C}_{r_i} ,\quad C_{r_i}=J-C_i, \: \hat{C}_{r_i}=C_{r_i}(\hat{u}_{-i},\tilde{u}_i)
 \end{equation}
Under this incentive scheme, the agents fully internalize the effect of their actions on the cost function of the central operator, in that the cost function of each agent becomes equal to $J$.
\begin{remark}
Note that the operator utilizes the knowledge of  the cost functions $C_i$. Further, the incentive design is similar to VCG in that the incentive for the $i$-th agent can be interpreted as the marginal cost incurred by the central operator due to the self-interested actions of the $i$-th agent. However, two factors make its form different from the traditional VCG design. First, the objective of the central operator is not social welfare given by the sum of the costs $C_{i}$, but rather, to minimize his cost $J$. Further, when the agent opts out, it still chooses an action that affects the cost incurred by the central operator. Second,  the actions are implemented by the agents and the operator does not have direct control over them. 
 \end{remark}


\begin{thm}
 \label{thm2}
The incentive design in \eqref{tax}:
\begin{enumerate}[(i)]
\item  is socially optimal if  Hessian of $J$ is positive definite for any $U$,
\item satisfies the participation constraint.
\item is weak budget balanced if 
\begin{equation}
\label{bb}
C_{r_i}(U^\star)-C_{r_i}(\hat{u}_{-i},\tilde{u}_i)\geq 0 ,\:\: \forall i.
\end{equation}
\end{enumerate}
 \end{thm}
 \begin{proof}
See Appendix in \cite{extend}.
 \end{proof}
  \begin{remark}
  The incentive scheme~(\ref{tax}) does not, in general,  satisfy  monotonicity  and equity.
   \end{remark}
 \begin{remark}
 It is worth mentioning that Theorem \ref{thm2}(i)  provides a sufficient condition on $J$ for achieving an optimal centralized solution in the distributed system. Further, the Hessian of $J$ being positive definite implies that there is a unique  optimal solution for the cost function of the operator. Thus, Theorem \ref{thm2}(i) states that if there is a unique solution that optimizes the cost function of the central operator, it can ensure that the agents will adopt those optimal actions by offering the incentive design specified in (\ref{tax}).
 \end{remark}
The price to pay for achieving social optimality and participation constraint in such an incentive design is that the incentives may be large and vary in a non-intuitive way among the agents. Thus, budget balance, monotonicity, and equity may be lost. 
\begin{example}
 Consider a system with the cost function of the central operator being $J=\frac{u^2_1}{2}+u^2_2-u_1+u_2-u_1u_2$ and two agents. It is easy to see that $U^\star=(1,0)$.
 \begin{itemize}
     \item If the cost functions of the agents are $C_1=\frac{u^2_1}{2}-u_1$ and $C_2=\frac{u^2_2}{2}+u_1u_2-u_2$, it is easy to calculate
     $(\hat{u}_{2},\tilde{u}_1)=(0,1)$, and $(\hat{u}_{1},\tilde{u}_2)=(1,0)$. It is easy to verify that \eqref{bb} holds and budget balance is satisfied.
     \item If the cost functions of the agents are $C_1=\frac{u^2_1}{2}+u_1$ and $C_2=\frac{u^2_2}{2}+u_1u_2-u_2$, budget balance is not satisfied.
     \item  Note that the incentive design in \eqref{tax} is socially optimal; hence,  $\theta_i=0$ for $i\in \mathcal N$. However, for $C_1=\frac{u^2_1}{2}+u_1$ and $C_2=\frac{u^2_2}{2}+u_1u_2-u_2$, although $\theta_1=\theta_2=0$, $t_1\neq t_2$. Thus, monotonicity is not satisfied in this case.
 \end{itemize}
\end{example}



\subsection{Discussion}

Table \ref{table1} summarizes the fulfillment of different properties by incentive designs under various assumptions on the system as analyzed in this paper. In this table,
`Y' means the allocation satisfies the corresponding property, whereas `C' states that satisfaction of the property is conditional.
\begin{table}
\vspace{3mm}
\centering
\caption{Comparison between  incentive schemes in terms of properties}
    \begin{tabular}{ p{2.7cm} | p{2cm} | p{2cm}  p{2cm}}
    \hline\hline
    \textbf{Property} & \textit{Known cost functions } & \textit{Private cost functions} \\ \hline
    Social optimality&Y & C \\ \hline
    Budget balance & C & Y\\
    \hline
      Participation constraint& Y & C\\
    \hline
      Equity $\&$ Monotonicity  & C& Y\\
    \hline   
    \end{tabular}
    \label{table1}\vspace{-6mm}
\end{table}
  
We see that the information available to the operator about the cost functions of the agents is crucial in selecting the properties that he can satisfy through an incentive scheme. If the cost functions are public information, then the operator can enact a VCG-like mechanism to ensure that the agents fully internalize the cost of their actions on his cost function. In other words, he can ensure that the actions chosen by the agents minimize his cost function. However, these incentives may not satisfy properties such as equity and monotonicity. On the other hand, if these cost functions are private information, satisfaction of social optimality and participation constraints may not be guaranteed.


 \section{Conclusion}
 \label{conclusion}

 In this paper, we considered a general distributed system with multiple agents who select and then implement actions on a system. The system has an operator with a given objective. However, the agents are self-interested and strategic in the sense that each agent optimizes its own individual objective.  The operator aims to mitigate this misalignment by designing an incentive scheme for the agents. The problem is difficult due to the cost functions of the agents being coupled, the objective of the operator not being social welfare, and the operator having no direct control over actions being implemented by the agents.
 
 This problem has been studied in many fields, particularly in mechanism design and cost allocation. However, mechanism design typically assumes that the objective of the operator is social welfare and the actions being implemented by the operator. On the other hand, cost allocation classically assumes that agents do not anticipate the effect of their actions on the incentive that they obtain. We have removed these assumptions and presented an incentive rule for this setup by bridging the gap between mechanism design and classical cost allocation. We have analyzed whether the proposed design satisfies various desirable properties such as social optimality, budget balance, participation constraint, and so on. We have also analyzed which of these properties can be  satisfied if the assumptions of cost functions of the agents being private and the agents being anticipatory are relaxed. 
 
 The work can be extended in various directions. The framework was inspired by smart infrastructure systems and various applications of the framework with formulations inspired by such systems can be considered. We are also interested in spanning the spectrum between full knowledge and no knowledge of the cost functions of the agents by the operator.

\bibliography{reference}
\bibliographystyle{ieeetr}


\end{document}